\documentclass[journal,onecolumn]{IEEEtran}
\ifCLASSINFOpdf
\else
\fi
\usepackage{cite}
\usepackage{amsmath,amssymb,amsfonts}
\usepackage{algorithmic}
\usepackage{graphicx}
\usepackage{textcomp}
\usepackage{xcolor}
\usepackage{lipsum}
\usepackage{tikz}
\usepackage{amsthm}
\newtheorem{theorem}{Theorem}

\newtheorem{remark}{Remark}
\theoremstyle{definition}
\newtheorem{definition}{Definition}
\usepackage{commath}
\usepackage{mathtools}
\graphicspath{ {./images/} }
\def\BibTeX{{\rm B\kern-.05em{\sc i\kern-.025em b}\kern-.08em
    T\kern-.1667em\lower.7ex\hbox{E}\kern-.125emX}}
\DeclarePairedDelimiterX\MeijerM[3]{\lparen}{\rparen}%
{\begin{smallmatrix}#1 \\ #2\end{smallmatrix}\delimsize\vert\,#3}

\newcommand\MeijerG[8][]{%
  G^{\,#2,#3}_{#4,#5}\MeijerM[#1]{#6}{#7}{#8}}

\begin{document}
%
\title{Performance Analysis of Multiple-Antenna Ambient Backscatter Systems at Finite Blocklengths\\}
%
%
%

\author{Likun Sui, \emph{Student Member}, Zihuai Lin, \emph{Senior Member}, Pei Xiao, \emph{Senior Member}, \tabularnewline
H. Vincent Poor, \emph{IEEE Fellow}, Branka Vucetic, \emph{IEEE Fellow}
\thanks{This research is supported by Australian Research Council (ARC) Discovery
projects DP190101988.
Likun Sui, Zihuai Lin and Branka Vucetic are with the School of Electrical and Information Engineering,
The University of Sydney, Australia (e-mail: \{likun.sui,zihuai.lin, branka.vucetic\}@sydney.edu.au).

Pei Xiao is with the Institute for Communication Systems (ICS), University of Surrey, UK (e-mail: p.xiao@surrey.ac.uk )

H. Vincent Poor is with the School of Electrical and Computer Engineering,
Princeton University, USA (e-mail: poor@princeton.edu).}}

\maketitle

\begin{abstract}
This paper analyzes the maximal achievable rate for a given blocklength and error probability over a multiple-antenna ambient backscatter channel with perfect channel state information at the receiver. The result consists of a finite blocklength channel coding achievability bound and a converse bound based on the Neyman-Pearson test and the normal approximation based on the Berry-Esseen Theorem. Numerical evaluation of these bounds shows fast convergence to the channel capacity as the blocklength increases and also proves that the channel dispersion is an accurate measure of the backoff from capacity due to finite blocklength. 
\end{abstract}


%
\IEEEpeerreviewmaketitle

\section{Introduction}\label{intro}

%
%
%
%

\IEEEPARstart{T}{he} Internet of Things (IoT) has drawn considerable attention from both academia and industry. However, as the demands for IoT increases dramatically, the provisioning of power to massive numbers of devices becomes a significant challenge. 
The novel concept of using passive communication techniques to enable communications for low-power devices is known as ambient backscatter \cite{b1,b2,b3a,b3b,b3c}, 
which is an RF power transfer technique \cite{RF_energy1,RF_energy2,RF_energy3,RF_energy4}. 
This approach offers a promising solution for communications between batteryless devices and it will enable the future growth of IoT systems.\par
Ambient backscatter communications use RF signals to transmit information symbols and harvest energy, resulting in battery-free operations.
The basic operating principles of ambient backscatter communications are as follows:
\begin{enumerate}
    \item A tag transmits a symbol, either $1$ or $-1$,  by backscattering and modulating a radio frequency (RF) signal from an existing ambient source.
    \item The receiver receives both the signal from the ambient source and the backscattered signal from the tag.
\end{enumerate}\par
Compared with active radio protocols, such as Wi-Fi, Bluetooth and ZigBee, an Ambient Backscatter Communication (AmBC) system has a relatively limited data rate. To mitigate this problem, multiple antenna techniques can be used to increase the date rate. For example, the authors of \cite{b23} use multiple antenna-based orthogonal frequency division multiplexing (OFDM) to design an AmBC system and cope with different channel conditions. Many current research works have focused on analyzing achievable rate and capacity for AmBC systems in an infinite blocklength regime \cite{b24}\cite{b25}. However, in practice, 
it is critical to evaluate how to maintain the desired error probability at a given finite blocklength. \par
The fundamental theorem of reliable data transmission limits over a noisy channel in terms of the mutual information $I(X;Y)$ between input $X$ and output $Y$ has been established in \cite{b14}.  The relationship between the data transmission rate and error probability has been demonstrated by various bounds in the finite blocklength regime \cite{b6}. In these bounds, the so-called information density plays an essential role, which is defined as
\begin{align*}
    i(X;Y)&=\log{\frac{dP_{XY}}{d(P_X\times P_Y)}(X,Y)} = \log{\frac{dP_{Y|X=x}}{dP_Y}(y)}
\end{align*}\par
In this paper, 
we consider an AmBC system with  multiple antennas on both the RF source and receiver sides and a single antenna tag. We analyze achievability and converse bound for the system together with normal approximations.
We first study a multiple input and multiple output (MIMO) channel of the AmBC system from the information-theoretic point of view. For the channel model, we consider the channel between the receiver and the RF source and the channel through the tag as a whole \cite{b24,b16}. 
We derive a corresponding achievability bound, which is defined as a lower bound on the size of a code that can be guaranteed to exist with a given arbitrary blocklength and error probability. And we further establish a corresponding converse bound, which is an upper bound on the size of any code with a given arbitrary blocklength and error probability. We then demonstrate that the channel capacity $C$ for the studied AmBC system in a finite blocklength regime is characterized by the channel dispersion $V$\cite{b21}, which is defined as a parameter to assess the stochastic variability of the channel relative to a deterministic channel with the same capacity \cite{b6,b9}. 

 

\subsection{Contributions}
\begin{itemize}

\item  We use binary hypothesis testing as a fundamental basis to provide achievability and converse bounds on the maximal achievable rate $R(n,\epsilon)$ for an AmBC system with multiple transmit and receive antennas. We consider the case when the transmitter and receiver have full channel state information (CSI) and hence we can perform waterfilling power allocation. 

\item Furthermore, 
to complete our achievability bound, we utilize the characteristic function to show that the output distribution is complex Gaussian.
Then we use mathematical methods to calculate the lower bounded $\kappa_{\tau}$ in (\ref{eq1}).

\item For the converse bound, we utilize the Mellin transform and Meijer G-function to obtain an upper bound on the auxiliary channel which is a product of $m$ copies of the probability density function (PDF) of Gamma distributed variables and apply Lebesgue measure to upper bound its output space.\par

\item  For comparison with the asymptotic performance of our achievability and converse bounds, we provide a normal approximation, which approximates the relationship among the rate $R(n,\epsilon)$, the channel capacity $C$ and the channel dispersion $V$.

\item In addition, we use the Berry-Esseen theorem to complete the proof. There is a $\frac{1}{\sqrt{n}}$ rate penalty for our AmBC system when compared to the Shannon channel capacity. 
\end{itemize}

\subsection{Notation}

We use lowercase letters to represent scalars. For a scalar complex value $x$, we denote by $|x|$, $\Re\{x\}$ and $\Im\{x\}$ its modulus, real part and imaginary part, respectively. A bold uppercase letter such as $\mathbf{X}$ denotes a random vector and its realization is represented by a bold lowercase symbol such as $\mathbf{x}$. Uppercase letters of special fonts are used to denote matrices, such as $\mathsf{Y}$ denotes a deterministic matrix and  $\mathbb{Y}$ represents a random matrix. We use $\mathbf{I}_a$ to denote the identity matrix of size $a\times a$. The superscript $\mathbb{A}^{H}$ denotes the Hermitian transposition of a matrix $\mathbb{A}$. $tr(\mathbb{A})$ denotes the trace of the matrix of $\mathbb{A}$. $\mathcal{CN}(\mu,\sigma^2)$ represents a complex Gaussian distribution with a mean of $\mu$ and a variance of $\sigma^2$; in particular, a complex Gaussian random variable $X\sim\mathcal{CN}(0,\sigma^2)$ with independent and identically distributed zero mean Gaussian real and imaginary components is circularly symmetric. $\norm{\mathbb{A}}$ stands for the Frobenius norm of a matrix $\mathbb{A}$, which is $\norm{\mathbb{A}}=\sqrt{tr(\mathbb{A}\mathbb{A}^{H})}$. $\mathbb{R}_{+}$ stands for the nonnegative real line; in particular, $\mathbb{R}_{+}^{m}$ is the nonnegative orthant of the $m$-dimensional real Euclidean spaces. $\log(\cdot)$ denotes the natural logarithm. We use $P_{Y|X}:\mathcal{A}\mapsto\mathcal{B}$ to denote a conditional probability measure between input and output spaces $\mathcal{A}$ and $\mathcal{B}$. $\mathbb{E}[x]$ denotes the statistical expectation of $x$ and $\mathbb{P}[A]$ denotes the probability of an event $A$.

The rest of this paper is organized as follows. Section \ref{sys} describes the system model and reviews the definition of a channel code with perfect channel estimation at the receiver. Section \ref{ach} derives the achievability bound for our system model. Section \ref{con} presents the converse bound for the investigated AmBC system. Section \ref{nor} gives the channel dispersion and the normal approximation of the studied system. Section \ref{tag} shows the relationship between the error probability of the source and tag signals. Numerical results are presented in Section \ref{num}. Finally, Section \ref{conclu} concludes the paper.

\section{System Model}\label{sys}

\begin{figure}[!t]
    \centering
    \includegraphics[width=2.5in]{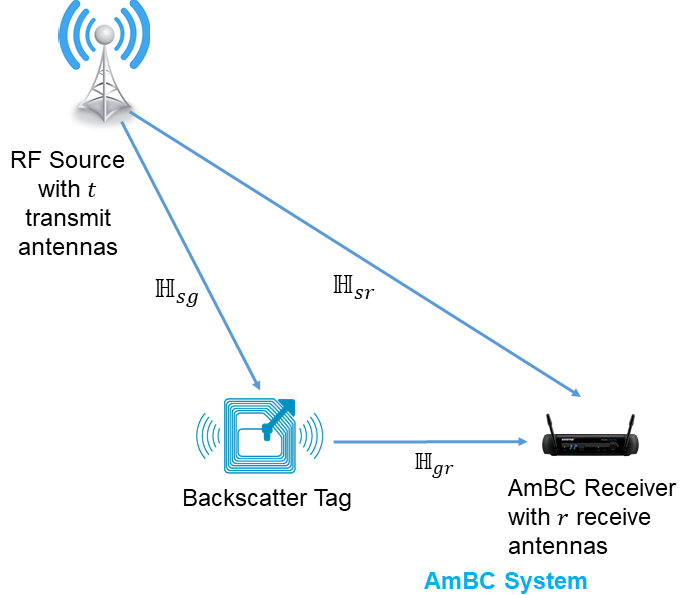}
    \caption{System model for ambient backscatter communications.}
    \label{fig:sys}
\end{figure}
We consider a MIMO ambient backscatter communication system with one RF source, one receiver and one backscatter tag with no battery as depicted in Fig. \ref{fig:sys}. The RF source and the receiver have $t$ and $r$ antennas, respectively and the tag has a single antenna.  We let $m\overset{\Delta}{=}\min(t,r)$ and denote by $\mathbb{H}_{sg}$, $\mathbb{H}_{gr}$ and $\mathbb{H}_{sr}$ the channel coefficient matrices between the source and tag, the tag and receiver, the source and receiver, respectively. And $\mathbb{H}_{sg}\in\mathbb{C}^{t\times 1},\mathbb{H}_{gr}\in\mathbb{C}^{1\times r}$,$\mathbb{H}_{sr}\in\mathbb{C}^{t\times r}$\cite{b13}. \par
A part of the tag received signal will be harvested to power the circuit of the tag, the rest would be backscattered to accomplish $``1"$ and $``-1"$ transmission. Without loss of generality, we assume tag's symbol remains unchanged for one block data transmission from the source. We denote the tag's symbol as $d\in\{-1, 1\}$. At the receiver, the received signal $\mathbb{Y}$ can be expressed by
\begin{equation}
    \mathbb{Y}=\mathbb{X}(\mathbb{H}_{sr}+\mathbb{H}_{sg}\mathbb{H}_{gr}Ad)+\mathbb{W},
\end{equation}
where $\mathbb{X}\in\mathbb{C}^{n\times t}$ is the signal transmitted over $n$ channel uses; $\mathbb{Y}\in\mathbb{C}^{n\times r}$ is the corresponding received signal; channel coefficients $\mathbb{H}_{sg}$, $\mathbb{H}_{gr}$ and $\mathbb{H}_{sr}$ are random but remain constant over the $n$ channel uses\cite{b7}\cite{b8}. $\mathbb{W}\in\mathbb{C}^{n\times r}$ is the additive noise at the receiver, which is independent of $\mathbb{H}_{sg}$, $\mathbb{H}_{gr}$ and $\mathbb{H}_{sr}$ and has independent and identically distributed (i.i.d) $\mathcal{CN}(0,1)$ entries; $A$ is the scattering efficiency of the tag.\par
Now let us denote $\mathbb{H}_0=\mathbb{H}_{sr}$ and $\mathbb{H}_1=A\mathbb{H}_{sg}\mathbb{H}_{gr}$, then  
\begin{equation}
\mathbb{Y}=\left\{
    \begin{array}{lr}
    \mathbb{X}(\mathbb{H}_0-\mathbb{H}_1)+\mathbb{W}, & d=-1,  \\
    \mathbb{X}(\mathbb{H}_0+\mathbb{H}_1)+\mathbb{W}, & d=1, 
    \end{array}
\right.    
\end{equation}\par
Next, we define a channel code with perfect channel state information at the receiver and transmitter \cite{b10}:\par
\begin{definition}
    an $(n,M,\epsilon)$ code consists of:
    \begin{enumerate}
    \item an encoder $f_t:\{1,\dots,M\}\times\mathbb{C}^{t\times r}\mapsto\mathbb{C}^{n\times t}$ that maps the message $j\in \{1,\dots,M\}$ and the channel $\mathsf{H}$ to a codeword $\mathsf{X}=f_t(j,\mathsf{H})$ which satisfies:
    \begin{align}
        \norm{\mathsf{X}}^2&=\norm{f_t(j,\mathsf{H})}^2\leq nP, \quad \forall j=1,\dots, M, \quad \forall \mathsf{H}\in\mathbb{C}^{t\times r}
    \end{align}
    \item a decoder $h_r:\mathbb{C}^{n\times r}\times\mathbb{C}^{t\times r}\mapsto\{1,\dots,M\}$ which satisfies:
    \begin{equation}
        \max_{1\leq j\leq M}{\mathbb{P}[h_r(\mathbb{Y},\mathbb{H})\neq J|J=j]\leq\epsilon}
    \end{equation}
    \end{enumerate}
\end{definition}

\section{Achievability Bound}\label{ach}

In this section, we assume the case where channel side information (CSI) is available at both the transmitter and receiver\cite{b10}. Our achievability bound is based on the $\kappa\beta$ bound which was recently provided in \cite{b6}. The achievability bound for the investigated AmBC system is given below and the corresponding proof builds on the characteristic of parallel AWGN channels.
\begin{theorem}\label{theorem1}
Let $g_1\geq\dots\geq g_m$ be the $m$ largest eigenvalues of a channel matrix, and let $\mathbf{g}$ represent $[g_1,\dots,g_m]^{T}$.
\begin{equation}
\mathbf{g}=\left\{
    \begin{array}{lr}
    \lambda_{max}\Big((\mathbb{H}_0-\mathbb{H}_1)^H(\mathbb{H}_0-\mathbb{H}_1)\Big),& d=-1,\\
    \lambda_{max}\Big((\mathbb{H}_0+\mathbb{H}_1)^H(\mathbb{H}_0+\mathbb{H}_1)\Big),& d=1,
    \end{array}
\right.  
\end{equation}
where $\lambda_{max}(\cdot)$ is a function which calculates the $m$ largest eigenvalues.\par
We consider power constraint with the water-filling strategy\cite{b11,b12}, and the following constraint set:
\begin{equation}\label{eq2}
    F_n\overset{\Delta}{=}\{\mathsf{X}: p_j(\mathsf{X})= [\lambda-\frac{1}{g_j}]^{+}\},
\end{equation}
where $\lambda$ is the solution of 
\begin{equation}
    \sum_{j=1}^{m} [\lambda-\frac{1}{g_j}]^{+} = P.
\end{equation}\par
Then the distributions of the information density under $P_{\mathbb{Y}}$ and under $P_{\mathbb{Y}|\mathbb{X}=\mathsf{X}}$ are:
\begin{equation}\label{eq3}
    G_n(\mathbf{p},\mathbf{g}) = \sum_{i=1}^{n}\sum_{j=1}^{m}\Big(\log{(1+g_jp_j(\mathsf{X}))}+1-\big| \sqrt{g_jp_j(\mathsf{X})}Z_{i,j}-\sqrt{1+g_jp_j(\mathsf{X})} \big| ^{2} \Big)
\end{equation}
and
\begin{equation}\label{eq4}
    H_n(\mathbf{p},\mathbf{g}) = \sum_{i=1}^{n}\sum_{j=1}^{m}\Big(\log{(1+g_jp_j(\mathsf{X}))}+1-\frac{\big| \sqrt{g_jp_j(\mathsf{X})}Z_{i,j}-1\big|^{2}}{1+g_jp_j(\mathsf{X})}\Big)
\end{equation}
respectively, where $Z_{i,j},i=1,\dots,n,j=1,\dots,m$ are i.i.d. $\mathcal{CN}(0,1)$ distributed random variables. For every $n$ and $0<\epsilon<1$, we have
\begin{equation}
    R(n,\epsilon|d)\geq \frac{1}{n}\log{\frac{\kappa_{\tau}}{\beta_{1-\epsilon+\tau}(P_{\mathbb{Y}},P_{\mathbb{Y}|\mathbb{X}=\mathsf{X}}|d)}}, 
\end{equation}
where $\beta_{1-\epsilon+\tau}(P_{\mathbb{Y}},P_{\mathbb{Y}|\mathbb{X}=\mathsf{X}}|d)$ can be defined by
\begin{equation}
    \beta_{1-\epsilon+\tau}(P_{\mathbb{Y}},P_{\mathbb{Y}|\mathbb{X}=\mathsf{X}}|d)=\mathbb{P}[G_n\geq\gamma_n|d],
\end{equation}
where $\gamma_n$ is chosen to satisfy
\begin{equation}
    \mathbb{P}[H_n\geq\gamma_n|d]=1-\epsilon+\tau.
\end{equation}\par
Further, we can lower-bound $\kappa_{\tau}$ by, 
\begin{equation}\label{eq35}
    \kappa_{\tau}\geq \tau/C_1,
\end{equation}
where $C_1$ is a constant.\par
We then obtain the achievability bound in a simple form.
\begin{equation}\label{eq1}
    R(n,\epsilon|d)\geq \frac{1}{n}\log{\frac{\tau/C_1}{\beta_{1-\epsilon+\tau}(P_{\mathbb{Y}},P_{\mathbb{Y}|\mathbb{X}=\mathsf{X}}|d)}} 
\end{equation}\par
Due to the conditional probability, we have
\begin{equation}
    R(n,\epsilon)\geq\frac{\mathbb{P}[d=-1]}{n}\log{\frac{\tau/C_1}{\beta_{1-\epsilon+\tau}(P_{\mathbb{Y}},P_{\mathbb{Y}|\mathbb{X}=\mathsf{X}}|d=-1)}}+\frac{\mathbb{P}[d=1]}{n}\log{\frac{\tau/C_1}{\beta_{1-\epsilon+\tau}(P_{\mathbb{Y}},P_{\mathbb{Y}|\mathbb{X}=\mathsf{X}}|d=1)}}
\end{equation}
\end{theorem}
The proof of (\ref{eq1}) can be found below.\par
\begin{proof}

To apply $\kappa\beta$ bound to the investigated AmBC channel model, we need to complete three steps: the first step is to choose the auxiliary output distribution $P_{\mathbb{Y}}$, the second step is to compute $\beta_{\alpha}$ and the third step is to compute $\kappa_{\tau}$.
\subsection{Choosing The Output Distribution}
$P_{\mathbb{Y}}$ can be expressed as a product of the distributions as
\begin{equation}
    P_{\mathbb{Y}} = \prod_{i=1}^{n}\prod_{j=1}^{m} P_{Y_{j,i}}.
\end{equation}\par
First, without loss of generality, to simplify calculation, we choose $\Re\{X_{j,i}\}=\Im\{X_{j,i}\}=\sqrt{P/2m}$. Therefore $Y_{j,i}$ can be expressed as
\begin{equation}
    Y_{j,i} = X_{j,i}(h_{sr}+h_{sg}h_{gr}Ad)+w,
\end{equation}
where $h_{sr}$, $h_{sg}$ and $h_{gr}$ are the element of $\mathbb{H}_{sr}$, $\mathbb{H}_{sg}$ and $\mathbb{H}_{gr}$ respectively and $w$ is an element of $\mathbb{W}$.\par

In the following, we separately analyze the real and imaginary parts of $Y_{j,i}$.

For the real part of $Y_{j,i}$, due to the property of the characteristic function, we have,
\begin{equation}
    \phi_{Y_{j,i}}(t) = \mathbb{E} [\exp\{it\Big(\sqrt{\frac{P}{2m}}(h_{sr}+h_{sg}h_{gr}Ad)+w\Big)\}]
    =e^{-\frac{Pt^2}{8}}e^{-\frac{t^2}{4}}\mathbb{E} [\exp\{it\sqrt{\frac{P}{2m}}h_{sg}h_{gr}Ad\}]\label{eq45}.
\end{equation}\par
Now we need to calculate the characteristic function of $h_{sg}h_{gr}$. At first, we need to derive the PDF of a product of two standard Gaussian distributed variables which is\cite{b28}
\begin{equation}\label{eq46aa}
    q_{1}(x)=\frac{2K_0(2|x|)}{\pi},
\end{equation}
where $K_0(\cdot)$ denotes a modified Bessel function of the second kind.\par
Then we can obtain the characteristic function of (\ref{eq46aa}) \cite{b29}
\begin{equation}
    \phi_{1}(t)=\int_{-\infty}^{\infty} e^{itx}q_{1}(x) dx
    =\int_{-\infty}^{\infty} e^{itx}\frac{2K_0(2|x|)}{\pi}dx
    =\frac{2}{\sqrt{(t^2+4)}}.
\end{equation}
\par
Thus we can obtain the characteristic function of $h_{sg}h_{gr}$
\begin{equation}
    \phi_{h_{sg}h_{gr}}(t)=\phi_{1}(t)\phi_{1}(-t)
    =\frac{4}{t^2+4}.\label{eq46}
\end{equation}\par
Then,
\begin{equation}
    \phi_{\sqrt{\frac{P}{2m}}h_{sg}h_{gr}Ad}(t)=\phi_{h_{sg}h_{gr}}(\sqrt{P/2m}Adt)
    =e^{-log(\frac{P}{8m}A^2d^2t^2+1)},
\end{equation}\par
and,
\begin{align}
    &\phi_{Y_{j,i}}(t)=e^{-\frac{Pt^2}{8}}e^{-\frac{t^2}{4}}e^{-\log(\frac{P}{8m}A^2d^2t^2+1)}\\
    &=e^{-\frac{(P+2)t^2}{8}}e^{-[\log(t_0)+\frac{2t_0}{1+t_0^2}(t-t_0)+\frac{2-2t_0^2}{(1+t_0^2)^2}(t-t_0)^2+\mathcal{O}(t^3)]}\label{eq47}\\
    &\sim e^{-\frac{1}{2}\sigma_{Y_{j,i}}^2t^2},
\end{align}
where (\ref{eq47}) comes from the Taylor expansion at $t=t_0$.\par

For the imaginary part of $Y_{j,i}$, it is the same as the real part. Thus ${\mathbb{Y}}$ is a random matrix whose $(j,i)$th entry follows a complex Gaussian distribution with variances being determined by $\sigma_{j}$,
\begin{equation}
    P_{Y_{j,i}} = \mathcal{CN}(0,\sigma_{j}^{2}\mathbf{I}_m),
\end{equation}
where $\sigma_{j}^{2}$ will be chosen later.
\subsection{Compute $\beta_{\alpha}$}
Due to the water-filling strategy, the constraint set can be found directly in (\ref{eq2}). The information density is shown below \cite{b6}:
\begin{equation}
    i(\mathbb{X}, \mathbb{Y}) = \sum_{j=1}^{m}\log{\frac{dP_{\mathbb{Y}|\mathbb{X}=\mathsf{X}}}{dP_{\mathbb{Y}}}(\mathbf{y})} = \sum_{j=1}^{m}\Big(\frac{n}{2}\log{\sigma_{j}^2} + \frac{1}{2}\sum_{i=1}^{n}\Big[\frac{y_i^2}{\sigma_{j}^2}-(y_i-\sqrt{g_jp_j(\mathsf{X})})^2\Big]\Big)
\end{equation}\par
It is convenient to define independent standard Complex Gaussian variables $Z_i\sim\mathcal{CN}(0,1)$. Then under $P_{\mathbb{Y}}$, the distribution of the information density is:
\begin{equation}
    G_n = \sum_{j=1}^{m}\Big(n\log{\sigma_{j}}+\frac{1}{2}\sum_{i=1}^{n}\Big((1-\sigma_{j}^2)|Z_i|^2+
    2\sqrt{g_jp_j(\mathsf{X})}\sigma_{j}|Z_i|-\frac{g_jp_j(\mathsf{X})}{2}\Big)\Big)
\end{equation}
and under $P_{\mathbb{Y}|\mathbf{X}=\mathbf{x}}$, the distribution is:
\begin{equation}
    H_n =\sum_{j=1}^{m}\Big(n\log{\sigma_{j}}+\frac{1}{2\sigma_{j}^2}\sum_{i=1}^{n}\Big((1-\sigma_{j}^2)|Z_i|^2+
    2\sqrt{g_jp_j(\mathsf{X})}|Z_i|+g_jp_j(\mathsf{X})\Big)\Big)
\end{equation}\par
The next step is to choose $\sigma_{j}^2$ which can minimize $\beta_\alpha$. Note that due to the equations $\beta_\alpha=\exp\{-D(P_{\mathbb{Y}|\mathbb{X}=\mathsf{X}}\parallel P_{\mathbb{Y}})\}$ and $D(P_{\mathbb{Y}|\mathbb{X}=\mathsf{X}}\parallel P_{\mathbb{Y}}) = \mathbb{E}[H_n]$, the minimization problem can be modified to the maximization of $\mathbb{E}[H_n]$. Then we have
\begin{equation}\label{eq16}
    \sigma_{j}^2 = 1+g_jp_j(\mathsf{X}),\quad  j=1,\dots,m.
\end{equation}
The proof of the above equation is shown in Appendix \ref{A}. With this choice of $\sigma_{j}^2$, we can get (\ref{eq3}) and (\ref{eq4}).

\subsection{Computing $\kappa_{\tau}$}
To compute $\kappa_\tau$, we need the same aforementioned power constraint set. We denote $\mathbf{U}$ as $P_{\mathbb{Y}}$ on $\mathbb{R}_{+}^{m}$,
\begin{equation}
    U_j = \sum_{i=1}^{n} (1+g_jp_j(\mathsf{X}))|Z_i|^2,
\end{equation}
where $U_j$ denotes the $j$th entry of $\mathbf{U}$, and $Z_i$ represents the $i$th i.i.d $\mathcal{CN}(0,1)$ distributed random variable. Given $\mathbf{p}$ and $\mathbf{g}$, the random variable $U_j$ is Gamma distributed. Thus its PDF $q_{U_j|\mathbf{p},\mathbf{g}}$ is
\begin{equation}\label{eq5}
    q_{U_j|\mathbf{p},\mathbf{g}}(r|p_j(\mathsf{X}),g_j)
    = \frac{2^n}{(1+g_jp_j(\mathsf{X}))^{n}\Gamma(n)}r^{n-1}\exp{(-\frac{2r}{1+g_jp_j(\mathsf{X})})}.
\end{equation}\par
Moreover, we denote $\mathbf{T}$ as $P_{\mathbb{Y}|\mathbb{X}=\mathsf{X}}$ on $\mathbb{R}_{+}^{m}$,
\begin{equation}
    T_j = \sum_{i=1}^{n} |Z_i+\sqrt{g_jp_j(\mathsf{X})}|^2, 
\end{equation}
where $T_j$ denotes the $j$th entry of $\mathbf{T}$. 
Given $\mathbf{p}$ and $\mathbf{g}$, the random variable $T_j$ is the sum of a non-central $\chi^2$ distributed random variable and a Gamma distributed random variable which can be expressed as:
\begin{equation}
    T_j = \frac{1}{4}K_j + J_j, 
\end{equation}
where $K_j\sim \chi'^2_n(4ng_jp_j(\mathsf{X}))$ and $J_j\sim \Gamma(n/2,1/2)$.\par
The characteristic function of $T_j$ can be written as
\begin{align}
    \phi_{T_j}(x)&=\phi_{\frac{1}{4}K_j+J_j}(x)=\mathbb{E}[\exp\{ix(\frac{1}{4}K_j+J_j)\}]\\  \nonumber 
    &=\mathbb{E}[\exp\{ix(\frac{1}{4}K_j)\}]\mathbb{E}[\exp\{ix(J_j)\}]\\\nonumber 
    &=\bigg[\frac{\exp\Big\{\frac{i4ng_jp_j(\mathsf{X})x}{1-2ig_jp_j(\mathsf{X})}\Big\}}{(1-2ix)^{n/2}}\bigg]^{1/4}\frac{1}{(1-i\frac{1}{2}x)^{n/2}}\\  \nonumber 
    &=\frac{\exp\Big\{\frac{ing_jp_j(\mathsf{X})x}{2-i2x}\Big\}}{(1-i2x)^{n/8}(1-i\frac{1}{2}x)^{n/2}}\sim \frac{\exp\Big\{\frac{ing_jp_j(\mathsf{X})x}{2-i2x}\Big\}}{(1-i2x)^{5n/8}}\\  \nonumber 
    &\sim \chi'^2_{5n/8}(ng_jp_j(\mathsf{X}))  \nonumber 
\end{align}\par
Thus the random variable $T_j$ obeys non-central $\chi^2$ distribution and its PDF $q_{T_j|\mathbf{p},\mathbf{g}}$ is
\begin{equation}\label{eq6}
    q_{T_j|\mathbf{p},\mathbf{g}}(r|p_j(\mathsf{X}),g_j)
    = \frac{1}{2}e^{-\frac{r+ng_jp_j(\mathsf{X})}{2}}(\frac{r}{ng_jp_j(\mathsf{X})})^{5n/32-1/2}I_{5n/16-1}(\sqrt{rng_jp_j(\mathsf{X})}),
\end{equation}
where $I_\alpha(y)$ is a modified Bessel function of a first kind.\par
We define,
\begin{equation}\label{eq7}
    f(r)\overset{\Delta}{=}\frac{ q_{T_j|\mathbf{p},\mathbf{g}}(r|p_j(\mathsf{X}),g_j)}{q_{U_j|\mathbf{p},\mathbf{g}}(r|p_j(\mathsf{X}),g_j)}.
\end{equation}\par
Substituting (\ref{eq5}) and (\ref{eq6}) into (\ref{eq7}), we have,
\begin{multline}\label{eq8}
    f(r) = (\frac{1}{2})^{n+1}(1+g_jp_j(\mathsf{X}))^n\Gamma(n)
    r^{-27n/32+1/2}(ng_jp_j(\mathsf{X}))^{-5n/32+1/2}\times\\
    I_{5n/16-1}(\sqrt{rng_jp_j(\mathsf{X})})\times\exp{\{\frac{1-g_jp_j(\mathsf{X})}{2(1+g_jp_j(\mathsf{X}))}r-\frac{ng_jp_j(\mathsf{X})}{2}\}}.
\end{multline}\par
We only consider the case where $5n/16$ is even, and if $5n/16$ is odd, we can replace $I_{5n/16-1}$ with $I_{5n/16-3/2}$. Thus the upper bound of modified Bessel function of its first kind $I_{5n/16-1}(x)$ is given by \cite{b26}:
\begin{multline}\label{eq9}
    I_{5n/16-1}(x)\leq \sqrt{\frac{\pi}{8x}}e^{x}(1+\frac{(5n/16-1)^2}{x^2})^{-1/4}\times\\
    \exp\Big\{-(5n/16-1)\sinh^{-1}{(\frac{5n/16-1}{x})}+x\Big(\sqrt{1+\frac{(5n/16-1)^2}{x^2}}-1\Big)\Big\}.
\end{multline}\par
Substituting $\sinh^{-1}(x)=\log(x+\sqrt{1+x^2})$ into (\ref{eq9}), we can obtain
\begin{multline}\label{eq38bb}
    I_{5n/16-1}(x)\leq
     \sqrt{\frac{\pi}{8}}\Big(x\sqrt{(1+\frac{(5n/16-1)^2}{x^2})}\Big)^{-1/2}\times\\\exp\Big\{-(5n/16-1)\log(\frac{5n/16-1}{x}+
    \sqrt{1+(\frac{5n/16-1}{x})^2})+x\Big(\sqrt{1+\frac{(5n/16-1)^2}{x^2}}\Big)\Big\}.
\end{multline}\par
Here we only consider that $r$ has the same order as $n$, which is $r=\mathcal{O}(n)$. Thus we set:
\begin{equation}\label{eq39}
    r=cn
\end{equation}
where $c\in[1+g_jp_j(\mathsf{X})-\delta, 1+g_jp_j(\mathsf{X})+\delta]$ for $\delta>0$.\par
For large $n$, we have
\begin{align}
    \sqrt{cn^2g_jp_j(\mathsf{X})}\sqrt{(1+\frac{(5n/16-1)^2}{cn^2g_jp_j(\mathsf{X})})} 
    &=\sqrt{cn^2g_jp_j(\mathsf{X})}\sqrt{1+\frac{25}{256cg_jp_j(\mathsf{X})}-\frac{5}{8cng_jp_j(\mathsf{X})}+\frac{1}{cn^2g_jp_j(\mathsf{X})}}\\
    &\sim \frac{n}{16}\sqrt{25+256cg_jp_j(\mathsf{X})}.\label{eq38aa}
\end{align}\par
Thus, putting (\ref{eq38aa}) into (\ref{eq38bb}), and after some algebraic manipulations, we obtain
\begin{multline}\label{eq38}
    I_{5n/16-1}(n\sqrt{cg_jp_j(\mathsf{X})})
    \leq\sqrt{\frac{\pi}{8}} \Big(\frac{n}{16}\sqrt{25+256cg_jp_j(\mathsf{X})}\Big)^{-1/2}\\ \exp\Big\{-(\frac{5n}{16}-1)\log\Big(\frac{\frac{5n}{16}-1+\frac{n}{16}\sqrt{25+256cg_jp_j(\mathsf{X})}}{n\sqrt{cg_jp_j(\mathsf{X})}}\Big)
    +\frac{n}{16}\sqrt{25+256cg_jp_j(\mathsf{X})}\Big\}.
\end{multline}\par
Using (\ref{eq8}) and (\ref{eq9}) and the expression of the Gamma function
\begin{equation}
    \log{\Gamma(n)} = n\log{n}-n-\frac{1}{2}\log{n}+\mathcal{O}(1),
\end{equation} 
$f(cn)$ can then be upper bounded by
\begin{equation}
    f(cn)\leq\exp\{-A(n,g_j,p_j(\mathsf{X}),c)+\mathcal{O}(1)\}.
\end{equation}
Here
\begin{multline}
    A(n,g_j,p_j(\mathsf{X}),c) = (\frac{5n}{16}-1)\log(5+\sqrt{25+256cg_jp_j(\mathsf{X})})-
    \frac{n\sqrt{25+256cg_jp_j(\mathsf{X})}}{16}+\frac{\log(\sqrt{25+256cg_jp_j(\mathsf{X})})}{2}-\\\frac{nc(1-g_jp_j(\mathsf{X}))}{2(1+g_jp_j(\mathsf{X}))}
    -(\frac{n}{4}-\frac{9}{2})\log{2}+(\frac{5n}{16}-1)\log{n}+(\frac{27n}{32}-\frac{1}{2})\log{c}-n\log{(1+g_jp_j(\mathsf{X}))}
    +\\(\frac{5n}{32}-\frac{1}{2})\log{g_jp_j(\mathsf{X})}+n(\frac{g_jp_j(\mathsf{X})}{2}+1)-\frac{1}{2}\log{\pi}.
\end{multline}\par
Since $c\in[1+g_jp_j(\mathsf{X})-\delta, 1+g_jp_j(\mathsf{X})+\delta]$ for $\delta>0$, we can estimate the range of $A(n,g_j,p_j(\mathsf{X}),c)$, which is
\begin{equation}
    0<A(n,g_j,p_j(\mathsf{X}),c)<2n.
\end{equation}\par
Thus, we can obtain the minimum value of $A(n,g_j,p_j(\mathsf{X}),c)$,
\begin{equation}
    A_{min}(n,g_j,p_j(\mathsf{X}),c)\geq 0.
\end{equation}\par
Therefore,  
\begin{equation}
    f(cn)\leq \exp\{-A_{min}(n,g_j,p_j(\mathsf{X}),c)+\mathcal{O}(1)\}=C_1.
\end{equation}\par
Without the loss of generality, we can assume a set $A$ where $\mathbb{P}_{\mathbf{T}}[A]\geq\tau$, then we reach the conclusion
\begin{align}
    \kappa_\tau &= \mathbb{P}_{\mathbf{U}}[A]= \int_{A}q_{\mathbf{U}|\mathbf{p},\mathbf{g}}(\mathbf{r}|\mathbf{p},\mathbf{g})d\mathbf{r}=\int_{A}\frac{q_{\mathbf{U}|\mathbf{p},\mathbf{g}}(\mathbf{r}|\mathbf{p},\mathbf{g})}{q_{\mathbf{T}|\mathbf{p},\mathbf{g}}(\mathbf{r}|\mathbf{p},\mathbf{g})}q_{\mathbf{T}|\mathbf{p},\mathbf{g}}(\mathbf{r}|\mathbf{p},\mathbf{g})d\mathbf{r}\nonumber\\
    &=\int_{A}\frac{1}{f(r)}q_{\mathbf{T}|\mathbf{p},\mathbf{g}}(\mathbf{r}|\mathbf{p},\mathbf{g})d\mathbf{r}\geq\frac{1}{C_1}\tau.
\end{align}
\end{proof}
\begin{remark}
    In practice, the numerical evaluation of the numerator of the right-hand side of (\ref{eq1}) is rather difficult. Inspired by Polyanskiy\cite{b6}\cite{b15}, finding an approximation of $\kappa_{\tau}$ is necessary. After doing so, the achievability bound can be computed and analyzed numerically. 
\end{remark}
\begin{remark}
    For the denominator of the right-hand side of (\ref{eq1}), the best way is the Monte Carlo technique. In our case, we normally set the repeat sampling to $10^5$ which can lead to an accurate result. 
\end{remark}
\section{Converse Bound}\label{con}
In this section, our converse bound is based on the meta-converse theorem\cite{b6}. Different from the aforementioned achievability bound, for the converse bound, we apply the power constraint of each codeword with equality. According to \cite{b22}, this power constraint is commonly used in the multiple antenna systems and in multiple access channels as well. For the convenience of further proof and analysis, we still use the eigenvalues to estimate each channel.
\begin{theorem}
    Let $g_1\geq\dots\geq g_m$ be the $m$ largest eigenvalues of channel matrix.
    \begin{equation}
    \mathbf{g}=\left\{
    \begin{array}{lr}
    \lambda_{max}\Big((\mathbb{H}_0-\mathbb{H}_1)^H(\mathbb{H}_0-\mathbb{H}_1)\Big),& d=-1,\\
    \lambda_{max}\Big((\mathbb{H}_0+\mathbb{H}_1)^H(\mathbb{H}_0+\mathbb{H}_1)\Big),& d=1,
    \end{array}
    \right.  
    \end{equation} 
    We consider that each codeword $\mathsf{X}$ satisfies the equal power constraint, and each $\mathsf{X}$ belongs to the set:
    \begin{equation}
        F'_n \overset{\Delta}{=} \{\mathsf{X}:\sum_{j=1}^m p_j(\mathsf{X})=P\}.
    \end{equation}\par
    For each codeword $\mathsf{X}\in F_n$, we have its corresponding power allocation vector, 
    \begin{equation}
        \mathbf{p}(\mathsf{X})\in\mathbb{R}^{m}: p_j(\mathsf{X})=\frac{1}{n}\norm{\mathsf{X}_{j,\cdot}}^2.
    \end{equation}
    There exists a constant $K>0$ such that for any $(M,\epsilon)$ code the maximal probability of error $\epsilon'$ over an auxiliary channel $Q$ satisfies
    \begin{equation}\label{eq10}
        1-\epsilon'\leq\frac{Kmn}{M}.
    \end{equation}\par
    Using $\beta_\alpha$ from Theorem. \ref{theorem1}, we can obtain the relationship between $\beta_\alpha$ and $\epsilon'$, 
    \begin{equation}\label{eq11}
        \inf_{\mathbf{p}(\cdot)}\beta_{1-\epsilon}(P_{\mathbb{Y}|\mathbb{X}=\mathsf{X}},Q_{\mathbb{Y}|\mathbb{X}=\mathsf{X}})\leq 1-\epsilon'.
    \end{equation}\par
    Combining (\ref{eq10}) and (\ref{eq11}), we can get the converse bound:
    \begin{equation}\label{eq12}
    R(n,\epsilon|d)\leq \frac{1}{n}\log\frac{Kmn}{\inf_{\mathbf{p}(\cdot)}\beta_{1-\epsilon}(P_{\mathbb{Y}|\mathbb{X}=\mathsf{X}},Q_{\mathbb{Y}|\mathbb{X}=\mathsf{X}}|d)}.
    \end{equation}\par
    Due to the conditional probability, we have
    \begin{equation}
        R(n,\epsilon)\leq\\ \frac{\mathbb{P}[d=-1]}{n}\log\frac{Kmn}{\inf_{\mathbf{p}(\cdot)}\beta_{1-\epsilon}(P_{\mathbb{Y}|\mathbb{X}=\mathsf{X}},Q_{\mathbb{Y}|\mathbb{X}=\mathsf{X}}|d=-1)}+\frac{\mathbb{P}[d=1]}{n}\log\frac{Kmn}{\inf_{\mathbf{p}(\cdot)}\beta_{1-\epsilon}(P_{\mathbb{Y}|\mathbb{X}=\mathsf{X}},Q_{\mathbb{Y}|\mathbb{X}=\mathsf{X}}|d=1)}.
    \end{equation}
\end{theorem}\par
The proof of (\ref{eq12}) can be found below.\par
\begin{proof}
    
We use the power allocation vector and the power constraint which can be found in Section \ref{ach}. 
The auxiliary channel $Q$ is defined as:
\begin{equation}
    Q_{\mathbb{Y}|\mathbb{X}=\mathsf{X}} \overset{\Delta}{=} \prod_{j=1}^{m}Q_{Y_{j}|\mathbb{X}=\mathsf{X}},
\end{equation}
where
\begin{equation}
    Q_{Y_{j}|\mathbb{X}=\mathsf{X}} = \mathcal{CN}(0,1+g_jp_{j}(\mathsf{X})).
\end{equation}\par
The output of the $Q_{\mathbb{Y}|\mathbb{X}=\mathsf{X}}$ channel is only dependent on $\mathbf{P}=\mathbf{p}(\mathbb{X})$. Let $\mathbf{S}=\mathbf{p}(\mathbb{Y})$ and its entries are the square of the norm of $\mathbb{Y}$ normalized by the blocklength $n$. It follows that $\mathbf{S}$ can express the statistics for the detection of $\mathbb{X}$ from $(\mathbb{Y},\mathbf{g})$. Thus we can define an equivalent channel $Q_{\mathbf{S}|\mathbf{P}}$ as 
\begin{equation}\label{eq49}
    S_j = \frac{1+g_jp_j(\mathsf{X})}{n}\sum_{i=1}^{n}|Z_{j,i}|^2,
\end{equation}
where $Z_{j,i}\sim\mathcal{CN}(0,1)$. Note that the random variable $S_j$ is Gamma distributed, and its PDF is given by
\begin{equation}
    q_{S_i|\mathbf{p}}(s_i) =
    \frac{(2n)^n}{(1+g_jp_j(\mathsf{X}))^n\Gamma(n)}s_i^{n-1}\exp\Big\{-\frac{2ns_i}{1+g_jp_j(\mathsf{X})}\Big\}.
\end{equation}\par
Moreover, as $Q_{\mathbf{S}|\mathbf{P}=\mathbf{p}}$ is a product of $m$ copies of the PDF of $S_i$ which is Gamma distribution. We can obtain the PDF of $Q_{\mathbf{S}|\mathbf{P}=\mathbf{p}}$ by the theorem shown below\cite{b28}.
\begin{theorem}\label{the2}
    Given $N$ independent gamma variables $x_i$ with the same shape parameter $k$ and the same scale parameter $\theta$ having density functions
    \begin{equation}\label{eq40}
        f_i(x_i)=\frac{1}{\Gamma(k)\theta^k}x_i^{k-1}e^{-\frac{x_i}{\theta}}.
    \end{equation}
    The probability density function $g(z)$ of the product $z=x_1x_2\dots x_N$ of $N$ independent gamma variables is a Meijer G-function multiplied by a normalising constant $\mathcal{K}$, 
    \begin{equation}
        g(z)=\mathcal{K}\MeijerG[\big]{N}{0}{0}{N}{}{k-1}{\frac{z}{\theta^N}},
    \end{equation}
    where
    \begin{equation}
        \mathcal{K}=(\frac{1}{\theta})^N\prod_{i=1}^{N}\frac{1}{\Gamma(k)},
    \end{equation}
    and
    \begin{equation}
        \MeijerG[\big]{m}{n}{p}{q}{j_1,j_2, \dots, j_p}{k_1,k_2, \dots, k_q}{z}
        =\frac{1}{2\pi i}\int_{c-i^\infty}^{c+i^\infty}z^{-s}\cdot\frac{\prod_{j=1}^{m}\Gamma(s+k_j)\cdot\prod_{j=1}^{n}\Gamma(1-j_j-s)}{\prod_{j=n+1}^{p}\Gamma(s+j_j)\cdot\prod_{j=m+1}^{q}\Gamma(1-k_j-s)}ds,
    \end{equation}
    where $c$ is a real constant defining a Bromwich path separating the poles of $\Gamma(s+k_j)$ from those of $\Gamma(1-j_k)-s$. 
\end{theorem}
\begin{proof}
    Since the Mellin integral transform of $\exp{(-\frac{x}{\theta})}$ in (\ref{eq40}) is 
    \begin{equation}
        \mathcal{M}\{\exp{(-\frac{x}{\theta})}|s\}=\int_0^\infty x^{s-1}e^{(-\frac{x}{\theta})}dx
        =\theta^s\int_0^\infty x^{s-1}e^{-x}dx
        =\theta^s\Gamma(s),
    \end{equation}
    and
    \begin{equation}
        M\{x^kf(x)|s\}=M\{f(x)|s+k\},
    \end{equation}
    it follows that the probability density function of the gamma variable in (\ref{eq40}) has the Mellin transform
    \begin{equation}
        \mathcal{M}\{f_i(x_i)|s\}=\int_0^\infty x^{s-1} \frac{1}{\Gamma(k)\theta^k}x^{k-1}e^{-\frac{x}{\theta}}dx=\frac{\theta^{s-1}}{\Gamma(k)}\int_0^\infty x^{(k+s-1)-1}e^{-x}dx=\frac{\Gamma(k+s-1)}{\Gamma(k)}\theta^{s-1}.
    \end{equation}\par
    Then the Mellin integral transform of the probability density function $g(z)$ of the product of $N$ independent random gamma variables is
    \begin{equation}
        \mathcal{M}\{g(z)|s\}=\prod_{i=1}^{N}\mathcal{M}\{f_i(x_i)|s\}=\theta^{N(s-1)}\prod_{i=1}^{N}\frac{\Gamma(k+s-1)}{\Gamma(k)},
    \end{equation}
    and
    \begin{align}
        g(z)&=\frac{1}{2\pi i}\int_{c-i^\infty}^{c+i^\infty} z^{-s}\theta^{N(s-1)}\prod_{i=1}^{N}\frac{\Gamma(k+s-1)}{\Gamma(k)}ds\nonumber\\
        &=\frac{1}{2\pi i\theta^N}\int_{c-i^\infty}^{c+i^\infty} (\frac{z}{\theta^N})^{-s}\prod_{i=1}^{N}\frac{\Gamma(k+s-1)}{\Gamma(k)}ds\nonumber\\
        &=(\frac{1}{\theta})^N\prod_{i=1}^{N}\frac{1}{\Gamma(k)}\MeijerG[\big]{N}{0}{0}{N}{}{k-1}{\frac{z}{\theta^N}}.
    \end{align}\par
    Thus we complete the proof.
\end{proof}
In our case, due to the shape parameter $k=n$ and the scale parameter $\theta=\frac{1+g_jp_j(\mathsf{X})}{2n}$ and the number of copies $N=m$, applying Theorem. \ref{the2}, we can derive the probability density function of $q_{\mathbf{S}|\mathbf{p}}$.
\begin{equation}\label{eq41}
    q_{S|\mathbf{p}}(z)=\mathcal{K}\MeijerG[\big]{m}{0}{0}{m}{}{n-1}{z(\frac{2n}{1+g_jp_j(\mathsf{X})})^m},
\end{equation}
where
\begin{equation}
    \mathcal{K}=(\frac{2n}{1+g_jp_j(\mathsf{X})})^m\prod_{i=1}^{m}\frac{1}{\Gamma(n)},
\end{equation}
and
\begin{equation}
    \MeijerG[\big]{m}{0}{0}{m}{}{n-1}{z(\frac{2n}{1+g_jp_j(\mathsf{X})})^m}
    =\frac{1}{2\pi i}\int_{c-i^\infty}^{c+i^\infty}(z(\frac{2n}{1+g_jp_j(\mathsf{X})})^m)^{-s}\prod_{j=1}^{m}\Gamma(s+n-1)ds.
\end{equation}\par
The PDF of $Q_{\mathbf{S}|\mathbf{p}}$ can be upper bounded by
\begin{equation}\label{eq42}
    q_{\mathbf{S}|\mathbf{p}}\leq K_1mn,
\end{equation}
where $K_1$ is a constant.\par
Because $\mathbf{p}(\mathsf{X})$ is the power allocation vector,  $\mathbf{p}$ belongs to a certain ball in $\mathbb{R}^{m}$. Due to the definition of $\mathbf{S}$, it belongs to a slightly larger ball. By definition of Lebesgue measure\cite{b20}, we can obtain the relationship between $Leb(\mathbf{p})$ and $Leb(\mathbf{S})$ and the bounded Lebesgue measure, 
\begin{equation}\label{eq17}
    Leb(\mathbf{p})\leq Leb(\mathbf{S}) \leq \frac{K_2}{M},
\end{equation}
where $Leb$ is Lebesgue measure and $K_2$ is a constant. Then for each codeword, the decoding set must have a Lebesgue measure smaller than $\frac{K_2}{M}$.\par
We have the lower bounds for both density of auxiliary channel and Lebesgue measure. Thus we obtain
\begin{equation}
    1-\epsilon'\leq \frac{Kmn}{M},
\end{equation}
where
\begin{equation}\label{eq30}
     K = K_1K_2.
\end{equation}
\end{proof}

\section{Normal Approximation}\label{nor}

In this section, we perform the asymptotic analysis of the maximum achievable rate $R(n,\epsilon)$ for a given blocklength $n$. According to normal approximation refinement of the coding theorem by Polyanskiy\cite{b6}, we define the channel dispersion as follows.
\begin{definition}
    The channel dispersion $V$ (measured in squared information units per channel use) of a given channel with channel capacity $C$ is given by
    \begin{equation}
        V=\lim_{\epsilon\rightarrow 0}\lim_{n\rightarrow\infty}{\sup}\frac{1}{n}\Big(\frac{nC-nR(n,\epsilon)}{Q^{-1}(\epsilon)}\Big)^2.
    \end{equation}
\end{definition}
Here, for the investigated AmBC system, the channel capacity
\begin{equation}
    C=\lim_{\epsilon\rightarrow 0}\lim_{n\rightarrow\infty}R(n,\epsilon)=\lim_{\epsilon\rightarrow 0}\lim_{n\rightarrow\infty}\Big(\mathbb{P}[d=1]R(n,\epsilon|d=1)+\mathbb{P}[d=-1]R(n,\epsilon|d=-1)\Big)=\sum_{j=1}^m \log{(1+g_jp_j(\mathsf{X}))}
\end{equation}
and the channel dispersion
\begin{equation}
    V=\sum_{j=1}^m\frac{g_jp_j(\mathsf{X})(g_jp_j(\mathsf{X})+2)}{(g_jp_j(\mathsf{X})+1)^2}=m-\sum_{j=1}^m \frac{1}{(1+g_jp_j(\mathsf{X}))^2}.
\end{equation}\par
The maximum achievable rate can be expressed by
\begin{equation}\label{eq15}
    R(n,\epsilon)=C-\sqrt{\frac{V}{n}}Q^{-1}(\epsilon)+\mathcal{O}(\frac{\log{n}}{n}).
\end{equation}\par
The proof of (\ref{eq15}) is shown in Appendix \ref{Appendix C}.
\section{The Analysis of The Tag Signal}\label{tag}

In this section, we analyze the tag's performance in terms of the achievability and the converse bound. Because the error probability of the RF source and the tag are mutually dependent, the performance of the tag can be derived based on the achievablility and the converse bound for the RF source.\par
Given $\mathbb{X}=\mathsf{X}$, $\mathbb{H}=\mathsf{H}$ and $\mathbb{Y}=\mathsf{Y}$, with maximum-likelihood (ML) detection and the estimated RF source signal $\widetilde{\mathsf{X}}$ ,$d$ can be estimated by:
\begin{equation}
    \widetilde{d} = \mathop{\arg\min}_{d}\norm{\mathsf{Y}-\widetilde{\mathsf{X}}\mathsf{H}_0-\widetilde{\mathsf{X}}\mathsf{H}_1d}.
\end{equation}\par
According to \cite{b30}, the search space in the above ML detection grows exponentially as modulation size increases, which causes extremely high complexity. For reducing the complexity of ML detection, we can apply maximum ratio combining (MRC) to the received signal $\mathsf{Y}$. To obtain the $\widetilde{d}$, we let
\begin{align}
    Z&=\Re\{\frac{\mathsf{H}_1^H\widetilde{\mathsf{X}}^H}{\norm{\mathsf{H}_1}^2}(\mathsf{Y}-\widetilde{\mathsf{X}}\mathsf{H}_0)\}\nonumber\\
    &=\Re\{\frac{\mathsf{H}_1^H\widetilde{\mathsf{X}}^H\mathsf{X}\mathsf{H}_1}{\norm{\mathsf{H}_1}^2}\}d+\Re\{\frac{\mathsf{H}_1^H(\widetilde{\mathsf{X}}^H\mathsf{X}-\widetilde{\mathsf{X}}^H\widetilde{\mathsf{X}})\mathsf{H}_0}{\norm{\mathsf{H}_1}^2}\}+\Re\{\frac{\mathsf{H}_1^H\widetilde{\mathsf{X}}^H\mathsf{W}}{\norm{\mathsf{H}_1}^2}\}\nonumber\\
    &=\Re\{\frac{\mathsf{H}_1^H\widetilde{\mathsf{X}}^H\mathsf{X}\mathsf{H}_1}{\norm{\mathsf{H}_1}^2}\}d+\Re\{\frac{\mathsf{H}_1^H(\widetilde{\mathsf{X}}^H\mathsf{X}-1)\mathsf{H}_0}{\norm{\mathsf{H}_1}^2}\}+\omega,
\end{align}
where $\omega\sim\mathcal{N}(0,1/(2\norm{\mathsf{H}_1}^2))$.\par
Due to ML detection, if $Z>0$, $\widetilde{d}=1$, otherwise, $\widetilde{d}=-1$.
\begin{enumerate}
    \item Case 1: when $\widetilde{\mathsf{X}}=\mathsf{X}$, we have
    \begin{equation}
        Z=d+\omega.
    \end{equation}
    It shows that $Z$ obeys normal Gaussian distribution. Therefore, under $\widetilde{\mathsf{X}}=\mathsf{X}$, the error probability of the event $\widetilde{d}\neq d$ can be expressed as below
    \begin{equation}
        \mathbb{P}[\widetilde{d}\neq d|\widetilde{\mathsf{X}}=\mathsf{X}]=(1-\epsilon)Q(\sqrt{2}\norm{\mathsf{H}_1}),
    \end{equation}
    where $Q(x)$ represents the Gaussian Q-function
    \begin{equation}
        Q(x)\overset{\Delta}{=}\int_{x}^{\infty}\frac{1}{\sqrt{2\pi}}e^{-t^2/2}dt.
    \end{equation}
    \item Case 2: when $\widetilde{\mathsf{X}}\neq\mathsf{X}$, we have
    \begin{equation}
        Z=-d+\Re\{\frac{\mathsf{H}_1^H(\widetilde{\mathsf{X}}^H\mathsf{X}-1)\mathsf{H}_0}{\norm{\mathsf{H}_1}^2}\}+\omega=-d-2\Re\{\frac{\mathsf{H}_1^H\mathsf{H}_0}{\norm{\mathsf{H}_1}^2}\}+\omega.
    \end{equation}
    Thus we can obtain
    \begin{equation}
        \mathbb{P}[\widetilde{d}
        \neq d|\widetilde{\mathsf{X}}\neq\mathsf{X}]\\=\epsilon\Big[\frac{1}{2}Q\Big(\sqrt{2}\norm{\mathsf{H}_1}(-1+2\Re\{\frac{\mathsf{H}_1^H\mathsf{H}_0}{\norm{\mathsf{H}_1}^2}\})\Big)
        +\frac{1}{2}Q\Big(\sqrt{2}\norm{\mathsf{H}_1}(-1-2\Re\{\frac{\mathsf{H}_1^H\mathsf{H}_0}{\norm{\mathsf{H}_1}^2}\})\Big)\Big].
    \end{equation}
\end{enumerate}\par
Once getting the error probability of the tag signal for both cases of $\widetilde{\mathsf{X}}=\mathsf{X}$ and $\widetilde{\mathsf{X}}\neq\mathsf{X}$, we have the equation which expresses the relationship between $\epsilon$ and the error probability of the tag signal, $\epsilon_d$.
\begin{equation}\label{eq48}
    \epsilon_d 
    = (1-\epsilon)Q(\sqrt{2}\norm{\mathsf{H}_1})\\+\epsilon\Big[\frac{1}{2}Q\Big(\sqrt{2}\norm{\mathsf{H}_1}(-1+2\Re\{\frac{\mathsf{H}_1^H\mathsf{H}_0}{\norm{\mathsf{H}_1}^2}\})\Big)
        +\frac{1}{2}Q\Big(\sqrt{2}\norm{\mathsf{H}_1}(-1-2\Re\{\frac{\mathsf{H}_1^H\mathsf{H}_0}{\norm{\mathsf{H}_1}^2}\})\Big)\Big].
\end{equation}\par
Therefore, we can insert (\ref{eq48}) into our achievability bound in Section \ref{ach} and the converse bound in Section \ref{con} to analyze the performance of the tag signal.
\section{Numerical Results}\label{num}

\begin{figure}[!t]
    \centering
    \includegraphics[width=4.5in]{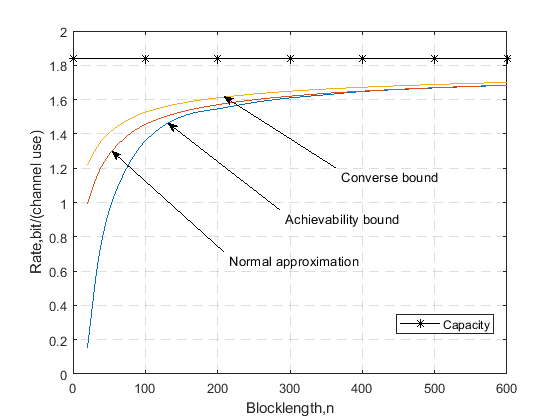}

    \caption{Achievability and converse bounds for $(n, M, \epsilon)$ codes over a MIMO ambient backscatter system model with Rayleigh fading channel and two transmit antennas and three receive antennas, SNR=$0$dB and $\epsilon=10^{-3}$.}
    \label{fig1}
\end{figure}

\begin{figure}[!t]
    \centering
    \includegraphics[width=4.5in]{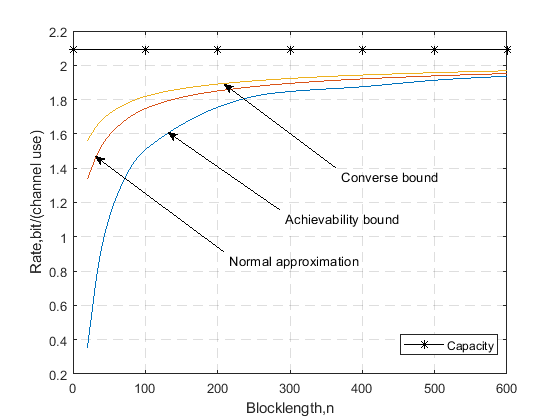}
    \caption{Achievability and converse bounds for $(n, M, \epsilon)$ codes over a MIMO ambient backscatter system model with Rician fading channel with $K$-factor equal to $10dB$ and two transmit antennas and three receive antennas, SNR=$0$dB and $\epsilon=10^{-3}$.}
    \label{fig2}
\end{figure}
\begin{figure}[!t]
    \centering
    \includegraphics[width=4.5in]{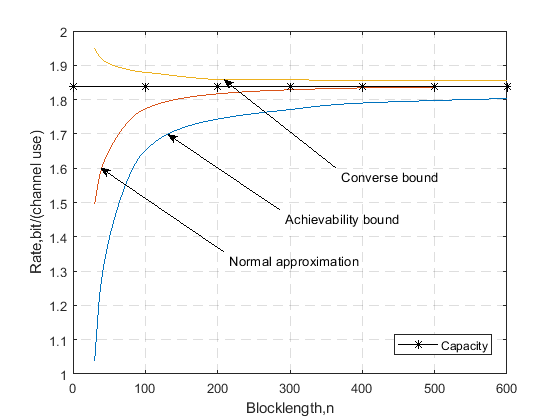}

    \caption{Achievability and converse bounds for $(n, M, \epsilon)$ codes over a MIMO ambient backscatter system model with Rayleigh fading channel and two transmit antennas and three receive antennas, SNR=$0$dB and the error probability of the tag signal $\epsilon_d=10^{-3}$ derived from (\ref{eq48}).}
    \label{fig3}
\end{figure}

In this section, we evaluate the bounds developed in Section \ref{ach} and Section \ref{con}. We consider an AmBC system consisting of a RF source with two transmit antennas, a batteryless tag with a single antenna and a receiver with three receive antennas. We assume all the channels, i.e., the channels between the source, the tag and the receiver, are independent and set the coefficient $A=0.5$, average error probability $\epsilon=10^{-3}$ and power $P=0$ dB. Fig. \ref{fig1} shows the numerical results of the derived bounds, the normal approximation and the Shannon capacity by assuming that all the channels are Rayleigh distributed.  From Fig. \ref{fig1}, we can see that $C=1.83$ bit/(channel use), and the blocklength required to achieve above $90\%$ of the capacity starts at $n=300$.\par

In Fig. \ref{fig2}, we change the channels from Rayleigh distributions to Rician distributions with $K$-factor equal to $10$dB and keep the same setting for other parameters as Fig. \ref{fig1}. We can see from Fig. \ref{fig2} that the blocklength required to achieve above $90\%$ of the capacity, i.e., $2.09$ bit/(channel use), starts at the blocklength $n=200$.\par

Fig. \ref{fig3} shows the results for the tag signals with the same setting as Fig. \ref{fig1} except that we fix the parameter $\epsilon_d=10^{-3}$ (please refer to (\ref{eq48}) for details). We observe that the gap between the channel capacity and the normal approximation is much smaller than the one presented in Fig. \ref{fig1}. When the blocklength is larger than $400$, the normal approximation can nearly reach the capacity. Fig. \ref{fig3} also implies that the convergence is much faster than Fig. \ref{fig1}.

\section{Conclusion}\label{conclu}
In this paper, we have established achievability and converse bounds on the maximal achievable rate $R(n,\epsilon)$ for a given blocklength $n$ and an average error probability $\epsilon$ in a multiple antenna backscatter system. We have also derived a normal approximation involving the channel dispersion and the channel capacity.\par
The analytical results  demonstrated that, as the blocklength $n$ increases, the channel dispersion $V$ implies rapid convergence to the channel capacity $C$. This suggests that the channel capacity is a valid performance metric for communication systems with stringent latency constraints operating over our system model. We have further developed an easy way to evaluate our approximation of $R(n, \epsilon)$ and demonstrated its accuracy using the achievability and the converse bounds. In our future work, we will consider using different coding schemes, such as \cite{distributedRaptor,Raptor_ML,JNCC,RCRC,codedcpm1,codedcpm2,codedcpm3}, in the AmBC system for applications in e.g., IoT sensor networks \cite{IoT_FD,WRN}, cellular networks \cite{cellular1,cellular2,cellular3,MIMO_capacity,network_capacity} and UAV networks \cite{UAVdownlink,UAV_THz,UAV_2} and derive corresponding bounds.


%

\appendices

\section{proof of (\ref{eq16})}\label{A}
In order to maximize $\mathbb{E}[H_n]$, we first need to obtain the exact expression of $\mathbb{E}[H_n]$ which is shown below,
\begin{align}
    \mathbb{E}[H_n] &= D(P_{\mathbb{Y}|\mathbb{X}=\mathsf{X}}\parallel P_{\mathbb{Y}})= D\Big(\mathcal{CN}(\sqrt{\mathbf{g}\mathbf{p}},\mathbf{I}_{m})\parallel \mathcal{CN}(0,\sigma_j^2\mathbf{I}_{m})\Big)\\
    &=\log{\det{\sigma_j^2\mathbf{I}_{m}}}-\log{\det{\mathbf{I}_{m}}}+(\sqrt{\mathbf{g}\mathbf{p}}-0)^H(\sigma_j^2\mathbf{I}_{m})^{-1}(\sqrt{\mathbf{g}\mathbf{p}}-0)+tr\Big((\sigma_j^2\mathbf{I}_{m})^{-1}(\mathbf{I}_{m})-\mathbf{I}_{m}\Big)\\
    &=m\log{\sigma_j^2}+\frac{\mathbf{g}\mathbf{p}}{\sigma_j^2}+m(\frac{1}{\sigma_j^2}-1).
\end{align}
Moreover, we need to calculate the derivative of $E[H_n]$ with respect of $\sigma_j^2$.
\begin{equation}
    \frac{d(\mathbb{E}[H_n])}{d\sigma_j^2} = \frac{d\Big(m\log{\sigma_j^2}+\frac{\mathbf{g}\mathbf{p}}{\sigma_j^2}+m(\frac{1}{\sigma_j^2}-1)\Big)}{d\sigma_j^2}
    =\frac{m}{\sigma_j^2}-\frac{\mathbf{g}\mathbf{p}}{(\sigma_j^2)^2}-\frac{m}{(\sigma_j^2)^2}.
\end{equation}
To determine the value of $\sigma_j^2$ which reach the peak point of $\mathbb{E}[H_n]$, we have to solve $\frac{d\mathbb{E}[H_n]}{d\sigma_j^2}=0$, then
\begin{equation}
     \sigma_{j}^2=1+g_jp_j(\mathsf{X}), \quad j=1,\dots,m.
\end{equation}

\section{Proof of (\ref{eq17})}\label{B}
In this section, we give the proof of bounded Lebesgue measure for $\mathbf{S}$. At first, due to space $\mathbf{p}$ which belongs to a certain ball in $\mathbb{R}^m$ with radius $P/m$, $\mathbf{S}$ belongs to a larger ball with the radius $r$
\begin{equation}
    r=\mathbb{E}[\frac{\sum_{i=1}^n|Z_{j,i}|^2}{n}]\Big(\mathbb{E}[g_j]\frac{P}{m}+1\Big)\leq\frac{1}{2}\Big(2m\frac{P}{m}+1\Big)\label{eq51}=P+\frac{1}{2},
\end{equation}
where (\ref{eq51}) follows from
\begin{align}
    &\mathbb{E}[g_j]=\frac{1}{m}trace\{(\mathbb{H}_0+Ad\mathbb{H}_1)^H(\mathbb{H}_0+Ad\mathbb{H}_1)\}\\
    &=\frac{1}{m}\Big(trace\{\mathbb{H}_0^H\mathbb{H}_0\}+2trace\{\mathbb{H}_0^H\mathbb{H}_1\}+trace\{\mathbb{H}_1^H\mathbb{H}_1\}\Big)\\
    &=\frac{1}{m}\Big(\sum_{i=1}^m|Z_i|^2+2\sum_{i=1}^m|Z_i|^3+\sum_{i=1}^m|Z_i|^4\Big)\\
    &\leq 2m.
\end{align}\par
We denote the $m$-dimensional ball in $\mathbb{R}^m$ by $B(m,r)$, where $n$ is the number of dimensions and $r$ is the radius of the ball. Thus, we can denote the volume of $B(m,1)$ by $s_m$. A ball with the radius $r$ has the volume
\begin{equation}
    vol_m(B(m,r))=s_mr^m.
\end{equation}\par
In order to evaluate $s_m$, we identify the ball in $\mathbb{R}^{m}$ with the Cartesian product $\mathbb{R}^2\times\mathbb{R}^{m-2}$. Furthermore, we apply the Fubini's theorem.
\begin{equation}
    s_m=vol_m(B(m,1))=\int_{B(m,1)}1d(x,y,z)
    =\int_{B(2,1)}\Big(\int_{B(m-2,\sqrt{1-x^2-y^2})}1dz\Big)d(x,y).
\end{equation}\par
In polar coordinates, we can obtain
\begin{equation}
    s_m=2\pi\int_{0}^{1}s_{m-2}(1-r^2)^{(m-2)/2}rdr,
\end{equation}
this equation leads to
\begin{equation}
    s_m=\frac{2\pi}{m}s_{m-2}.
\end{equation}\par
After recursion processing, we have
\begin{equation}
    s_m=\frac{\pi^{m/2}}{\Gamma(\frac{m}{2}+1)}.
\end{equation}\par
Therefore, 
\begin{align}
    vol_m(B(m,r))&=\frac{\pi^{m/2}}{\Gamma(\frac{m}{2}+1)}\times r^m\\
    &\leq\frac{\pi^{m/2}}{\Gamma(\frac{m}{2}+1)}\times (P+\frac{1}{2})^m\\
    &\leq \frac{\pi^{m/2}}{\sqrt{m}(\frac{m}{2})^{m/2}}\times (P+\frac{1}{2})^m\label{eq50}\\
    &\leq K_2,
\end{align}
where (\ref{eq50}) is obtained from Stirling's formula.

\section{Proof of (\ref{eq15})}\label{Appendix C}
It is convenient to separate the proof of (22) into two parts. At first, we address the converse part.
In the beginning of the proof, we need to introduce the following inequality\cite{b6}, for any $\gamma>0$,
\begin{equation}\label{eq18}
    1-\epsilon\leq \mathbb{P}[\frac{dP}{dQ}\geq\gamma]+\gamma\beta_{1-\epsilon}(P,Q).
\end{equation}
Using (\ref{eq18}) with $P=P_{\mathbb{Y}|\mathbb{X}=\mathsf{X}}$ and $Q=Q_{\mathbb{Y}|\mathbb{X}=\mathsf{X}}$, then we get for any $\eta$
\begin{equation}\label{eq29}
    \beta_{1-\epsilon}(P_{\mathbb{Y}|\mathbb{X}=\mathsf{X}},Q_{\mathbb{Y}|\mathbb{X}=\mathsf{X}})\geq e^{n\eta-nC}
    \Big(1-\epsilon-\mathbb{P}[\sum_{i=1}^n H_i(\mathbf{g},\mathbf{p})\leq \eta]\Big)
\end{equation}
with
\begin{equation}
    H_i=\sum_{j=1}^m \Big(\log{(1+g_jp_j(\mathsf{X}))}+1-\frac{|\sqrt{g_jp_j(\mathsf{X})}Z_{i,j}-1|^2}{1+g_jp_j(\mathsf{X})}\Big)
\end{equation}
where $Z_{i,j}\sim\mathcal{CN}(0,1)$.
To continue the proof, we need to introduce an important tool which is Berry-Esseen Theorem.
\begin{theorem}\label{the1}
    Let $X_k,k=1,\dots,n$ be independent with
    \begin{equation}
        \mu_k=\mathbb{E}[X_k],\quad
        \sigma^2=Var[X_k],\quad
        t_k=\mathbb{E}[|X_k-\mu_k|^3],\quad
        \sigma^2=\sum_{k=1}^{n}\sigma_k^2 \quad \textrm{and} \quad
        T=\sum_{k=1}^{n}t_k.
    \end{equation}
    Then for any $-\infty<
    \lambda<\infty$
    \begin{equation}
        \Big|\mathbb{P}\Big[\sum_{k=1}^{n}(X_k-\mu_k)\geq\lambda\sigma\Big]-Q(\lambda)\Big|\leq \frac{6T}{\sigma^3}
    \end{equation}
\end{theorem}
Since $Z_{i,j}$ is an i.i.d complex Gaussian random variable, the first moment of $H_i$ is given by
\begin{align}
    \mathbb{E}[H_i]&=\mathbb{E}[\sum_{j=1}^m \Big(\log{(1+g_jp_j(\mathsf{X}))}+1-\frac{|\sqrt{g_jp_j(\mathsf{X})}Z_{i,j}-1|^2}{1+g_jp_j(\mathsf{X})}\Big)]\\
    &=\sum_{j=1}^m \Big(\log{(1+g_jp_j(\mathsf{X}))}+1-\frac{\mathbb{E}[|\sqrt{g_jp_j(\mathsf{X})}Z_{i,j}-1|^2]}{1+g_jp_j(\mathsf{X})}\Big)\\
    &=\sum_{j=1}^m \Big(\log{(1+g_jp_j(\mathsf{X}))}+1-\frac{\mathbb{E}[g_jp_j(\mathsf{X})|\Re\{Z_{i,j}^2\}|-2\sqrt{g_jp_j(\mathsf{X})}|\Re\{Z_{i,j}\}|+1+|\Im\{Z_{i,j}\}|]}{1+g_jp_j(\mathsf{X})}\label{eq19}\\
    &=\sum_{j=1}^m \log{(1+g_jp_j(\mathsf{X}))}\\
    &=C
\end{align}
where (\ref{eq19}) from $|\Re\{Z_{i,j}^2\}|=1$ and $|\Re\{Z_{i,j}\}|=|\Im\{Z_{i,j}\}|=0$.\\
The variance of $H_i$,
\begin{equation}
    \sigma^2 = \mathbb{E}[(H_i-E[H_i])^2]=\mathbb{E}\Big[\Big(\sum_{j=1}^m 1-\frac{|\sqrt{g_jp_j(\mathsf{X})}Z_{i,j}-1|^2}{1+g_jp_j(\mathsf{X})}\Big)\Big]=m-\sum_{j=1}^m\frac{1}{(1+g_jp_j(\mathsf{X}))^2}=V
\end{equation}
Furthermore, from the Radon-Nikodym derivative between $P_{\mathbb{Y}|\mathbb{X}}$ and $Q_{\mathbb{Y}|\mathbb{X}}$, we can define
\begin{equation}\label{eq27}
    B = 6\frac{\mathbb{E}[|J|^3]}{(\mathbb{E}[|J^2|])^{\frac{3}{2}}},
\end{equation}
where
\begin{equation}
    J = \sum_{j=1}^{m} \log(1+g_jp_j(\mathsf{X})) + \frac{p_j(\mathsf{X})+2\sqrt{p_j(\mathsf{X})}Z_j-p_j(\mathsf{X})Z_j^2}{1+p_j(\mathsf{X})}.
\end{equation}
Noting the fact that $p_j(\mathsf{X})\leq P,j=1,\dots,m$, so we can bound $|J|$, as
\begin{equation}\label{eq20}
    |J|\leq \sum_{j=1}^{m} \log(1+g_jP) + \Big(P+2\sqrt{P}|Z_j|+PZ_j^2\Big).
\end{equation}
Noting that the right-hand side of (\ref{eq20}) is independent of the choice of $p_j(\mathsf{X})$, for any choice of $p_j(\mathsf{X})$ there exists constants $\zeta_1\geq 0$ and $\zeta_2\geq 0$ that make,
\begin{align}
    \mathbb{E}[|J|^2]&\leq\zeta_1\label{eq23}\\
    \mathbb{E}[|J|^3]&\leq\zeta_2\label{eq24}
\end{align}
Since the variance of $J$ is
\begin{equation}
    \mathbb{E}[J^2]=m- \sum_{j=1}^{m} \frac{1}{(1+g_jp_j(\mathsf{X}))^2},
\end{equation}
and as $\sum_{j=1}^{m}p_j(\mathsf{X})=P$, we have at least one $p_j(\mathsf{X})$ which is larger than $\frac{P}{m}$. Therefore we obtain the lower bound for $\mathbb{E}[J^2]$, which is
\begin{equation}\label{eq21}
    \mathbb{E}[J^2]\geq m- \frac{1}{(1+g_j\frac{P}{m})^2}.
\end{equation}
By the Lyapunov inequality\cite{b27}, we have
\begin{equation}\label{eq22}
    (\mathbb{E}[|J|^2])^{1/2}\leq (\mathbb{E}[|J|^3])^{1/3}.
\end{equation}
Combining (\ref{eq21}) and (\ref{eq22}) together, we have the lower bound for $\mathbb{E}[|J|^3]$ as well. Thus for any choice of $p_j(\mathsf{X})$ there are constants $\zeta_3>0$ and $\zeta_4>0$ that make
\begin{align}
    \mathbb{E}[|J|^2]&\geq\zeta_3\label{eq25}\\
    \mathbb{E}[|J|^3]&\geq\zeta_4\label{eq26}
\end{align}
Combining (\ref{eq23}), (\ref{eq24}), (\ref{eq25}) and (\ref{eq26}) into (\ref{eq27}) yields
\begin{equation}
    0<B<+\infty.
\end{equation}
For sufficiently large $n$, we have 
\begin{equation}\label{eq33}
    a_n = 1-\epsilon-\frac{2B}{\sqrt{n}}>0.
\end{equation}
Without loss of generality, we can also assume for such $n$, 
\begin{equation}\label{eq31}
    \eta=-\sqrt{nV}Q^{-1}(a_n).
\end{equation}
From Theorem. \ref{the1}, we obtain
\begin{equation}
    \mathbb{P}[\sum_{i=1}^n H_i\leq \eta]\leq a_n+\frac{B}{\sqrt{n}}\leq 1-\epsilon-\frac{B}{\sqrt{n}}\label{eq28}.
\end{equation}
Substituting (\ref{eq28}) into (\ref{eq29}), we have
\begin{equation}
    \beta_{1-\epsilon}(P_{\mathbb{Y}|\mathbb{X}=\mathsf{X}},Q_{\mathbb{Y}|\mathbb{X}=\mathsf{X}})\geq e^{n\eta-nC}\frac{B}{\sqrt{n}}.
\end{equation}
Using the estimation of $K(n)$ in (\ref{eq30}),
\begin{equation}
    \log{K(n)}=\log{n}+\mathcal{O}(1).
\end{equation}
We arrive at the upper bound as
\begin{equation}
    R(n,\epsilon)\leq C-\frac{\eta}{n}-\frac{1}{n}\log{(1-\mathbb{P}[\sum_{i=1}^n H_i\leq \eta]-\epsilon)}+\frac{\log{n}}{n}+\mathcal{O}(\frac{1}{n})\leq C-\eta + \frac{1}{n}\log{n} +\mathcal{O}(\frac{\log{n}}{n})
\end{equation}
Using Taylor's theorem, for the interval of $\theta\in \big[1-\epsilon-\frac{2B}{\sqrt{n}},1-\epsilon\big]$, from (\ref{eq31}), we have
\begin{equation}\label{eq32}
    \eta = -\sqrt{nV}Q^{-1}(1-\epsilon)+2\sqrt{\sigma^2}B\frac{d}{dx}Q^{-1}(\theta)
\end{equation}
Since the derivative of inverse Gaussian Q-function is a continuous function in $(0,1)$, we can obtain the upper bound of second term of the right-hand side of (\ref{eq32}),
\begin{equation}
    2\sqrt{\sigma^2}B\frac{d}{dx}Q^{-1}(\theta)\leq 0
\end{equation}
So, we have
\begin{equation}
    R(n,\epsilon)\leq C+\sqrt{\frac{V}{n}}Q^{-1}(1-\epsilon)+\mathcal{O}(\frac{\log{n}}{n})
\end{equation}
Due to the property of $Q^{-1}(x)$, which is for any $\epsilon\in(0,1)$ $Q^{-1}(1-\epsilon)=-Q^{-1}(\epsilon)$. Thus, we complete the proof
\begin{equation}
    R(n,\epsilon)\leq C-\sqrt{\frac{V}{n}}Q^{-1}(\epsilon)+\mathcal{O}(\frac{\log{n}}{n})
\end{equation}
Secondly, we address the achievability part.\\
We set 
\begin{equation}
    a_n=1-\epsilon+\frac{2B}{\sqrt{n}}.
\end{equation}
Note that for sufficiently large $n$, $a_n<1$, so it makes $\eta$ in (\ref{eq33}) meaningful.\\
Using Theorem. \ref{the1}, and combining with $\mathbb{P}[\sum_{i=1}^n H_i\leq \eta]-a_n<0$, we have
\begin{equation}
    \mathbb{P}[\sum_{i=1}^n H_i\leq \eta]\geq a_n-\frac{B}{\sqrt{n}}\geq 1-\epsilon+\frac{B}{\sqrt{n}}.
\end{equation}
For $\eta=-\sqrt{nV}Q^{-1}(a_n)$, without loss of generality, we can assume
\begin{equation}
    \log\gamma = nC-\eta=nC+\sqrt{nV}Q^{-1}(a_n)
\end{equation}
Therefore, we obtain
\begin{equation}
    P_{\mathbb{Y}|\mathbb{X}}[i(\mathsf{X};\mathbb{Y})\geq\log\gamma]=\mathbb{P}[\sum_{i=1}^n H_i\leq\eta]\geq 1-\epsilon+\frac{B}{\sqrt{n}}
\end{equation}
by setting
\begin{equation}
    \tau=\frac{B}{\sqrt{n}}
\end{equation}

Assuming $\sigma$ is the variance of $i(\mathsf{X};\mathbb{Y})$, $T$ is the third moment of $i(\mathsf{X};\mathbb{Y})$, we have
\begin{align}
    \beta_{1-\epsilon+\tau}\leq& \mathbb{E}[\exp\{-i(\mathsf{X};\mathbb{Y})\}1_{\{i(\mathsf{X};\mathbb{Y})\geq\log\gamma\}}]\\
    \leq& \sum_{c=0}^{\infty}\exp\{-(\log{(\frac{\sqrt{n}\gamma}{\sigma})}+c\delta)\}\mathbb{P}\Big[\log{(\frac{\sqrt{n}\gamma}{\sigma})}\nonumber+c\delta
    \leq i(\mathsf{X};\mathbb{Y})\leq\log{(\frac{\sqrt{n}\gamma}{\sigma})}+(c+1)\delta\Big]\\
    \leq& \sum_{c=0}^{\infty}\exp\{-(\log{(\frac{\sqrt{n}\gamma}{\sigma})}+c\delta)\}\bigg( \int_{(\log{(\frac{\sqrt{n}\gamma}{\sigma})}+c\delta)/\sigma}^{(\log{(\frac{\sqrt{n}\gamma}{\sigma})}+c\delta+\delta)/\sigma}\frac{1}{\sqrt{2\pi}}e^{-t^2/2}dt+\frac{12T}{\sigma^3}\bigg)\label{eq43}\\
    \leq& \sum_{c=0}^{\infty}\exp\{-(\log{(\frac{\sqrt{n}\gamma}{\sigma})}+c\delta)\}\frac{1}{\sigma}\Big(\frac{\delta}{\sqrt{2\pi}}+\frac{12T}{\sigma^3}\Big)\\
    \leq& \frac{2\Big(\frac{\delta}{\sqrt{2\pi}}+\frac{12T}{\sigma^3}\Big)}{\sqrt{n}\gamma}\label{eq44}
\end{align}
where (\ref{eq43}) is from Theorem. \ref{the1} and (\ref{eq44}) from choosing $\delta=\log{2}$ and $\sum_{c=0}^{\infty}2^{-c}=2$.\\
Then, we obtain
\begin{align}
    \log\beta_{1-\epsilon+\tau}\leq&-\log\gamma-\frac{1}{2}\log{n}+\mathcal{O}(1)\label{eq34}\\
    =&-nC+\eta-\frac{1}{2}\log{n}+\mathcal{O}(1)\label{eq36}
\end{align}
From (\ref{eq35}), we can choose $\kappa_{\tau}$,
\begin{equation}\label{eq37}
    \log\kappa_{\tau}\geq-\frac{1}{2}\log{n}+\mathcal{O}(1).
\end{equation}
Thus, by combining (\ref{eq36}) and (\ref{eq37}), we can conclude the achievability part
\begin{equation}
    R(n,\epsilon)\geq C+\sqrt{\frac{V}{n}}Q^{-1}(1-\epsilon)+\mathcal{O}(\frac{\log{n}}{n})=C-\sqrt{\frac{V}{n}}Q^{-1}(\epsilon)+\mathcal{O}(\frac{\log{n}}{n}).
\end{equation}\par
This completes the proof.



\ifCLASSOPTIONcaptionsoff
  \newpage
\fi


\begin{thebibliography}{1}

\bibitem{b1} L. Yan, Y. Zhang, L. T. Yang, and H. Ning, \textit{The Internet of Things:
From RFID to the Next-Generation Pervasive Networked Systems.}
Boca Raton, FL, USA: Auerbach Publications, 2008.
\bibitem{b2} H. Jayakumar, K. Lee, W. S. Lee, A. Raha, Y. Kim, and
V. Raghunathan, “Powering the Internet of Things,” in \textit{Proc. Int.
Symp. Low Power Electron. Design}, La Jolla, CA, USA, Aug. 2014,
pp. 375–380.
\bibitem{b3a}	Y. Hu, P. Wang, Z. Lin, M. Ding, YC. Liang, “Performance Analysis of Ambient Backscatter Systems with LDPC-coded Source Signals” in IEEE Transactions on Vehicular Technology, vol. 70, no. 8, pp. 7870-7884, Aug. 2021, doi: 10.1109/TVT.2021.3093912.

\bibitem{b3b}	Y. Hu, P. Wang, Z. Lin, M. Ding, YC. Liang, “Machine Learning Based Signal Detection for Ambient Backscatter Communications”, 2019 IEEE International Conference on Communications (ICC): SAC Internet of Things Track. 

\bibitem{b3c} Xing, Z. Lin, M. Ding,  “Outage Capacity Analysis for Ambient Backscatter Communication Systems”, 2018 28th International Telecommunication Networks and Applications Conference (ITNAC).



\bibitem{RF_energy1}J. Wang, B. Li, G. Wang, Z. Lin, H. Wang, and G. Chen, Optimal Power Splitting for MIMO SWIPT Relaying Systems with Direct Link in IoT Networks, Physical Communication, Volume 43, December 2020. 

\bibitem{RF_energy2}D. Zhai, H. Chen, Z. Lin. Y. Li and B. Vucetic, “Accumulate Then Transmit: Multi-user Scheduling in Full-Duplex Wireless-Powered IoT Systems”, IEEE Internet of Things Journal, Volume: 5 , Issue: 4 , Aug. 2018. 

\bibitem{RF_energy3} H. Chen, Y. Ma. Z. Lin, Y. Li and B. Vucetic, “Distributed Power Control in Interference Channels with QoS Constraints and RF Energy Harvesting: A Game-Theoretic Approach”, IEEE Transactions on Vehicular Technology, Volume: 65, Issue: 12, Dec. 2016. pp.10063 – 10069. 

\bibitem{RF_energy4} Y. Ma, H. Chen, Z. Lin, Y. Li and B. Vucetic, "Distributed and Optimal Resource Allocation for Power Beacon-Assisted Wireless-Powered Communications,", IEEE Transactions on  Communications, vol.63, no.10, pp.3569-3583, Oct. 2015. 

\bibitem{b6} Y. Polyanskiy, H. V. Poor, and S. Verd\'{u}, “Channel coding rate in the
finite blocklength regime,” \textit{IEEE Trans. Inf. Theory}, vol. 56, no. 5,
pp. 2307–2359, May 2010.
\bibitem{b7} D. N. C. Tse and P. Viswanath, \textit{Fundamentals of Wireless Communication}.
Cambridge, U.K.: Cambridge Univ. Press, 2005.
\bibitem{b8} G. J. Foschini and M. J. Gans, “On limits of wireless communications in
a fading environment when using multiple antennas,” \textit{Wireless Personal
Commun.}, vol. 6, no. 3, pp. 311–335, 1998.
\bibitem{b9} S. Verd\'{u} and T. S. Han, “A general formula for channel capacity,” \textit{IEEE Trans. Inf. Theory}, vol. 40, no. 4, pp. 1147–1157, Jul. 1994.
\bibitem{b10} W. Yang, G. Durisi, T. Koch, and Y. Polyanskiy, “Diversity versus
channel knowledge at finite block-length,” in \textit{Proc. IEEE Inf. Theory Workshop (ITW)}, Lausanne,
Switzerland, Sep. 2012, pp. 577–581.
\bibitem{b11} G. Caire, G. Taricco, and E. Biglieri, “Optimum power control
over fading channels,” \textit{IEEE Trans. Inf. Theory}, vol. 45, no. 5,
pp. 1468–1489, May 1999.
\bibitem{b12} \.{I}. E. Telatar, “Capacity of multi-antenna Gaussian channels,” \textit{Eur. Trans. Telecommun.}, vol. 10, pp. 585–595, Nov. 1999.
\bibitem{b13} E. Biglieri, J. Proakis, and S. Shamai (Shitz), “Fading channels:
Information-theoretic and communications aspects,” \textit{IEEE Trans. Inf.
Theory}, vol. 44, no. 6, pp. 2619–2692, Oct. 1998.

\bibitem{b14} C. E. Shannon, “Probability of error for optimal codes in a Gaussian
channel,” \textit{Bell Syst. Tech. J.}, vol. 38, no. 3, pp. 611–656, May 1959.
\bibitem{b15} Y. Polyanskiy, “Channel coding: Non-asymptotic fundamental limits,”
Ph.D. dissertation, Dept. Elect. Eng., Princeton Univ., Princeton, NJ,
USA, 2010.
\bibitem{b16} W. Yang, G. Durisi, T. Koch, and Y. Polyanskiy, “Quasi-static SIMO
fading channels at finite blocklength,” in \textit{Proc. IEEE Int. Symp. Inf. Theory (ISIT) 2013}, Istanbul,Turkey, Jul., pp. 1531–1535.

\bibitem{b17} J. Neyman and E. S. Pearson, “On the problem of the most efficient
tests of statistical hypotheses,” \textit{Philosoph. Trans. Roy. Soc. A}, vol. 231,
pp. 289–337, Jan. 1933.
\bibitem{b18} H. V. Poor and S. Verd\'{u}, “A lower bound on the error probability in multihypothesis
testing,” \textit{IEEE Trans. Inf. Theory}, vol. 41, no. 6, pp. 1992–1993, 1995.
\bibitem{b19} R. E. Blahut, “Hypothesis testing and information theory,” \textit{IEEE Trans. Inf. Theory},
vol. 20, no. 4, pp. 405–417, 1974.
\bibitem{b20} R. Ash, \textit{Information Theory.} New York: Interscience Publishers, 1965.
\bibitem{b21} Y. Polyanskiy, H. V. Poor, and S. Verd\'{u}, “Dispersion of Gaussian channels,” in \textit{Proc.
2009 IEEE Int. Symp. Inf. Theory (ISIT)}, Seoul, Korea, Jul. 2009.
\bibitem{b22} S. Verd\'{u}, “Spectral efficiency in the wideband regime,” \textit{IEEE Trans. Inf. Theory},
vol. 48, no. 6, pp. 1319–1343, Jun. 2002.
\bibitem{b23} G. Yang, Y.-C. Liang, R. Zhang, and Y. Pei, “Modulation in the air: Backscatter communication over ambient OFDM carrier,” \textit{IEEE Trans. Commun.}, vol. 66, no. 3, pp. 1219–1233, Mar. 2018.
\bibitem{b24} J. Qian, Y. Zhu, C. He, F. Gao and S. Jin, "Achievable Rate and Capacity Analysis for Ambient Backscatter Communications," \textit{IEEE Trans. Commun.}, vol. 67, no. 9, pp. 6299-6310, Sept. 2019
\bibitem{b25} D. Darsena, G. Gelli, and F. Verde, “Modeling and performance analysis of wireless networks with ambient backscatter devices,” \textit{IEEE Trans. Commun.}, vol. 65, no. 4, pp. 1797–1814, Jan. 2017.
\bibitem{b26} A. V. Prokhorov, “Inequalities for Bessel functions of a purely imaginary argument,” \textit{Theor. Probability Appl.}, vol. 13, pp. 496–501, 1968.
\bibitem{b27} T. Holliday, A. Goldsmith, and P. Glynn, “Capacity of finite state channels based on Lyapunov exponents of random matrices,” \textit{IEEE Trans. Inf. Theory}, vol. 52, no. 8, pp. 3509–3532, Aug. 2006.
\bibitem{b28}Springer, M. D., and W. E. Thompson. “The Distribution of Products of Beta, Gamma and Gaussian Random Variables.” \textit{SIAM Journal on Applied Mathematics}, vol. 18, no. 4, 1970, pp. 721–737.
\bibitem{b29}McNolty, Frank. “Some Probability Density Functions and Their Characteristic Functions.” \textit{Math. Comput.}, vol. 27, no. 123, 1973, pp. 495–504.
\bibitem{b30}G. Yang, Q. Zhang and Y. Liang, "Cooperative Ambient Backscatter Communications for Green Internet-of-Things," in \textit{IEEE Internet of Things J.}, vol. 5, no. 2, pp. 1116-1130, April 2018.


\bibitem{distributedRaptor}J Yue, Z Lin, B Vucetic, G Mao, T Aulin, "Performance analysis of distributed raptor codes in wireless sensor networks", IEEE Transactions on Communications 61 (10), 2013, 4357-4368

\bibitem{Raptor_ML}P Wang, G Mao, Z Lin, M Ding, W Liang, X Ge, Z Lin, "Performance analysis of raptor codes under maximum likelihood decoding", IEEE Transactions on Communications 64 (3), 2016, 906-917

\bibitem{JNCC}K Pang, Z Lin, Y Li, B Vucetic, "Joint network-channel code design for real wireless relay networks", the 6th International Symposium on Turbo Codes \& Iterative Information, 2010, 429-433.

\bibitem{RCRC}Z Lin, A Svensson, "New rate-compatible repetition convolutional codes",
IEEE Transactions on Information Theory 46 (7),  2651-2659

\bibitem{codedcpm1}Z. Lin and T. Aulin, “Joint Source and Channel Coding using Punctured Ring Convolutional Coded CPM”, IEEE Transactions on Communications, Vol. 56, No. 5, May, 2007, pp. 712-723. 
\bibitem{codedcpm2}Z. Lin and T. Aulin, “On Combined Ring Convolutional Coded Quantization and CPM for Joint Source and Channel Coding”, Transactions on Emerging Telecommunications Technologies, Special Issue on ’New Directions in Information Theory’, Vol.19, No.4. June 2008, pp. 443-453. 
\bibitem{codedcpm3} Z. Lin and T. Aulin, “Joint Source-Channel Coding using Combined TCQ/CPM: Iterative Decoding”, IEEE Transactions on Communications, VOL.53, NO. 12, Dec. 2005, pp. 1991-1995.

\bibitem{IoT_FD}D Zhai, H Chen, Z Lin, Y Li, B Vucetic, "Accumulate then transmit: Multiuser scheduling in full-duplex wireless-powered IoT systems", IEEE Internet of Things Journal 5 (4), 2753-2767

\bibitem{WRN}J. Yue; Z. Lin; B. Vucetic; G. Mao; M. Xiao; B. Bai; K. Pang, "Network Code Division Multiplexing for Wireless Relay Networks,"  IEEE Transactions on Wireless Communications, vol.14, no.10, pp.5736-5749, Oct. 2015.

\bibitem{cellular1}Z. Lin, P. Xiao and B. Vucetic, “Analysis of Receiver Algorithms for LTE SC-FDMA Based Uplink MIMO Systems”, IEEE Transactions on Wireless Communications, Vol. 9, No. 1, Nov. 2010, pp. 60-65. 

\bibitem{cellular2}Y Chen, M Ding, D Lopez-Perez, J Li, Z Lin, B Vucetic, "Dynamic reuse of unlicensed spectrum: An inter-working of LTE and WiFi", IEEE Wireless Communications 24 (5), 52-59

\bibitem{cellular3}Y Chen, J Li, Z Lin, G Mao, B Vucetic, "User association with unequal user priorities in heterogeneous cellular networks", IEEE Transactions on Vehicular Technology 65 (9), 7374-7388

\bibitem{MIMO_capacity}Z Lin, B Vucetic, J Mao, "Ergodic capacity of LTE downlink multiuser MIMO systems", 2008 IEEE International Conference on Communications, 3345-3349.

\bibitem{network_capacity}G Mao, Z Lin, X Ge, Y Yang, "Towards a simple relationship to estimate the capacity of static and mobile wireless networks", IEEE transactions on wireless communications 12 (8), 2014, 3883-3895	

\bibitem{UAVdownlink}D López-Pérez, M Ding, H Li, LG Giordano, G Geraci, A Garcia-Rodriguez, Z. Lin, M. Hassan, "On the downlink performance of UAV communications in dense cellular networks", 2018 IEEE global communications conference (GLOBECOM), 1-7

\bibitem{UAV_THz}	X. Wang, P. Wang, M. Ding, Z. Lin, L. Hanzo and B. Vucetic, “Performance Analysis of TeraHertz Unmanned Aerial Vehicular Networks”, in IEEE Transactions on Vehicular Technology, vol. 69, no. 12, pp. 16330-16335, Dec. 2020, doi: 10.1109/TVT.2020.3035831. 

\bibitem{UAV_2}C Liu, M Ding, C Ma, Q Li, Z Lin, YC Liang, "Performance analysis for practical unmanned aerial vehicle networks with LoS/NLoS transmissions", IEEE International Conference on Communications Workshops (ICC Workshops), 2018, 1-6.


\end{thebibliography}
\end{document}